\documentclass[11pt]{article}
\usepackage{dsfont,amsmath,amsthm,amssymb}
\usepackage{fullpage}

\usepackage{bbding,pifont}

\usepackage{amsmath,amssymb,amsthm}
\newtheorem{thm}{Theorem}

\newtheorem{lem}[thm]{Lemma}

\usepackage{tabularx}
\usepackage{graphicx}
\usepackage{stmaryrd}
\usepackage{caption}
\usepackage{subcaption}
\usepackage{tikz}
\usepackage{xfrac}
\usepackage{multirow}
\usetikzlibrary{arrows.meta, positioning, petri}
\newcommand{\eqdef}{\stackrel{\text{def}}{=}}

\newcommand{\dG}{\mathbb{G}}
\newcommand{\IN}{\mathbb{N}}

\newcommand{\lcm}{{\mathrm{lcm}}}

\newcommand{\rem}[2]{ \langle {#1} \rangle_{\!#2}}

\newcolumntype{s}{>{\hsize=.5\hsize}X}

\newcommand{\Cd}{{\mathcal{C}}}
\newcommand{\Cbi}{{\mathcal{C}}_{\mathrm{sym}}}
\newcommand{\Cst}{{\mathcal{C}}^*}
\newcommand{\Cstbi}{{\mathcal{C}}^*_{\mathrm{sym}}}

\newcommand{\mbset}{ \{ \!\! \{ }
\newcommand{\meset}{ \} \!\! \} }
\newcommand{\bigmbset}{\big \{ \!\! \big\{ }
\newcommand{\bigmeset}{\big \} \!\! \big\}}

\newcommand{\ignore}[1]{}

\usepackage{authblk}
 
\title{Clock Synchronization Is Almost Impossible with Bounded Memory}
\date{}
\author{Bernadette Charron-Bost\thanks{Corresponding author: \texttt{charron@di.ens.fr}}}
\author{Louis Penet de Monterno}
\affil[1]{DI ENS, CNRS, \'Ecole Normale Sup\'erieure, 75230 Paris, France}
\begin{document}
\maketitle

\begin{abstract}
We study the problem of clock synchronization in a  networked system with arbitrary starts 
	for all nodes.
We consider a synchronous network of $n$ nodes, where each node has a local clock that is an integer counter.
Eventually, clocks must be all equal and increase by  one in each round modulo some period $P$.   
The purpose of this paper is to study whether  clock synchronization can be achieved with bounded memory,
	that is every node maintains a number of states that does not depend on the network size.
In particular, we are interested in clock synchronization algorithms which work in dynamic networks, i.e.,
	tolerate that communication links continuously fail and come-up.
	
We first focus on  self-stabilizing solutions for clock synchronization, and  prove that there is no such algorithm 
	that is bounded memory, even in the case of static networks.
More precisely, we show a lower bound of $n+1$ states at each node required to achieve 
	clock synchronization  in static strongly connected networks with at most $n$ nodes, 
	and derive a lower bound of $n-2$ rounds on synchronization time, in the worst case.
We then prove that, when the self-stabilizing requirement is removed, the impossibility of clock synchronization 
	with bounded memory still holds in the dynamic setting: every solution for the clock synchronization problem
	in dynamic networks with at most $n$ nodes requires each node to have $\Omega(\log n)$ states.
\end{abstract}

\newpage
\section{Introduction}

In this paper, we study the following \emph{clock synchronization problem}:
In a system of $n$ agents receiving common regular pulses, where each agent is equipped with a 
	discrete clock, all these local clocks must eventually increase by one at each pulse and be equal, 
	modulo some fixed integer~$P$, despite arbitrary activation times or arbitrary initializations.
This problem, also referred to as \emph{synchronous counting}~\cite{DolevHJKLRSW:jcss:2016},
	corresponds to synchronization in \emph{phase}, as opposed to  synchronization in \emph{frequency} 	
	in non-synchronous systems (e.g., see~\cite{ST87,ER95,FL04,LLW10}).

Clock synchronization is a fundamental problem, both in engineered and natural systems:
For engineered systems, it is often required that processors have a common notion of time  
	(e.g., see the universal self-stabilizing algorithm in~\cite{BV:dc:02} or the very many distributed 
	algorithms structured into synchronized \emph{phases},  like the 2PC and 3PC protocols, or the 
	consensus protocols in~\cite{DLS:jacm:88}).
Clock synchronization also corresponds to an ubiquitous phenomenon in the natural world and finds 
	numerous applications in physics and biology, e.g., the Kuramoto model for the synchronization 
	of coupled oscillators~\cite{Str:phys:00}, synchronous flashing fireflies, 
	collective synchronization of pancreatic beta cells~\cite{Jad:cacm:12}.  	 
In such systems with massive amounts of cheap, bulk-produced hardware, or swarms of small mobile 
	agents that are tightly constrained by the systems they run on, the resources available at each 
	agent may be severely limited. 	 
In this context, \emph{bounded memory algorithms}, i.e.,  algorithms where agents  have a  finite set 
	of states that is independent of the number $n$ of agents, are highly desirable.
For example, this is typically the case for the agents  in the \emph{Population Protocol} 
	model~\cite{AngluinAER:dc:2007}.
In bounded memory algorithms, agents de facto send only $O(1)$-bit messages.
Moreover, since the  number of agents is arbitrary, agents may have no identifier.
	
The main goal of this work is to study whether the bounded memory limitation is crippling or not for 
	clock synchronization. 
For that, we consider a classical abstract model that captures the common requirements of the various 
	settings cited above.
A networked system is composed of $n$ identical agents that have \emph{deterministic} behaviors.
They  operate in synchronous rounds and communicate by simple broadcast. 
In order to take into account link failure and link creation, communication links may vary over time.
The basic connectivity assumption is of a finite diameter,  i.e., any pair of agents can communicate, 
	possibly indirectly, over a period of time that is uniformly bounded throughout the execution.
Agents may use only local informations; in particular they  are unaware of the structure of the network 
	and of  its diameter. 

\vspace{0.2cm}
\noindent {\bf Contribution.} 
We answer the above question of whether the bounded memory condition is crippling or not for 
	the clock synchronization problem in the affirmative in two cases.
We first  consider \emph{self-stabilizing} solutions, i.e., tolerating arbitrary initializations.
We show that there  exists no  such solutions for the clock synchronization problem using bounded memory, 	
	 even in the case of static networks with permanent communication links.
More precisely, we show that $n+1$ states at each node are required to achieve 
	clock synchronization  in the class of static strongly connected networks with at most $n$ nodes.
We also derive a lower bound of $n-2$ rounds on synchronization time, in the worst case.
These lower bounds demonstrate that the self-stabilizing clock synchronization algorithms proposed 
	in~\cite{BPV:algorithmica:08,FelKS:spaa:20,CBPM:opodis:22} when an upper bound on the number 
	of agents is available,  are all asymptotically optimal with regards to both runtime and memory
	requirements.

Our  second result extends this impossibility result to non-self-stabilizing solutions:
	we show that clock synchronization with bounded memory is impossible in the dynamic setting.
Indeed, we prove that clock synchronization in dynamic networks of size  at most $n$ 
	 requires each agent to have $\Omega(\log n)$ states. 

\vspace{0.2cm}
\noindent {\bf Related works}.
Clock synchronization has been extensively studied in different communication models, under different 
	assumptions, and with multiple variants in the problem specification.
The pioneer papers on this question~\cite{ADG:ppl:91,HG:ipl:95,BPV:algorithmica:08,ADDP19}, which all
	proposed self-stabilizing solutions for static networks,  assume that a bound on the network size 
	or on the diameter is available.
Only the synchronization algorithms in~\cite{BV:dc:02} and~\cite{CBPM:opodis:22} for static and dynamic
	networks, respectively,  dispenses with this assumption.
Both algorithms are self-stabilizing and use \emph{finite but unbounded} memory, 
	in the sense that the number of states used by each agent  in any given execution is finite but is 
	not uniformly bounded with respect to the number of agents.

There are also numerous works on clock synchronization in the context of \emph{Byzantine} failures
	 (e.g., see~\cite{Dol:97,DW:jacm:04, LenzenRS:podc:2015,LenzenRS:siamcomp:2017}). 
The goal is to tolerate one third of Byzantine agents, which obviously increases runtime and memory 
	requirements.
The clock synchronization algorithms deployed for tolerating Byzantine failures typically assume a  fully 
	connected network and port labellings, which requires each agent to use at least $\log n$ bits.

Synchronization has also been studied in the model of \emph{Population Protocols}~\cite{AngluinAER:dc:2007}.
As in our model, there is a finite but unbounded number of  anonymous agents, with bounded memory.
The difference lies in the interaction mode between agents: this is an asynchronous model, 
	where a unique random pair of  agents interact in each step.
The synchronization task is totally different as it consists in simulating synchronous rounds
	(synchronization in frequency).
Moreover,  only stabilization with probability one or with high probability is required, 
	while our algorithms have to stabilize in \emph{all} executions.

The same probabilistic weakening of problem specification is considered for other  probabilistic 
	communication models.
In particular, in the \emph{PULL} model~\cite{KDG:focs:03}, Bastide et al.~\cite{BGS:soda:21} proposed a 
	self-stabilizing algorithm that  synchronizes $P$-periodic clocks in $O(\log n)$ rounds with high 
	probability, in the case  $P$ is a power of 2.
This demonstrates how randomness in communications can help to circumvent our impossibility results 
	and lower bounds for clock synchronization. 

The clock synchronization problem has also been studied in the round-based 
	\emph{Beeping model}~\cite{CK:disc:10},
	which is a very restrictive communication model that relies on carrier 
	sensing: in each round, an agent can either listen  or transmit  to all its neighbors one beep.
In this model,  for synchronizing $P$-periodic clocks, Feldmann et al.~\cite{FelKS:spaa:20} proposed a 
	(non self-stabilizing) solution with $4P$ state agents.
This contrasts but does not contradict our second impossibility result since their algorithm assumes
	static symmetric networks (that is, static bidirectional communication links) and  the restricted model 
	of \emph{diffusive starts}, where a passive node  becomes active upon the receipt of a message from an 
	active incoming neighbor.
In this context, they also proved lower bounds on the number of states and the runtime of self-stabilizing
	solutions which are of the same order than ours, namely
	 $n$ (instead of $n+1$) states at each node are required and  synchronization takes at least $n$ 
	 (instead of $n-2$) rounds in the class of static symmetric networks with at most $n$ nodes.	
Our lower bounds hence show that relaxing communication constraints -- in particular, allowing to send 
	more than one bit per message --
	does not  help to reduce memory requirements and  runtime for clock synchronization.

\section{Preliminaries}\label{sec:model}
 
\subsection{The computing model} \label{sec:alg}

We consider a networked system with a fixed and finite set of nodes denoted $0,\cdots,n-1$
	that may be \emph{passive} or \emph{active}:
A passive node sends nothing, receives nothing, and does not change state,
	as opposed to an active node which can send messages, receives messages, 
	and update its state.
An active node remains active forever.
All nodes are initially passive and eventually become active.

Computation proceeds in  \emph{synchronized rounds}, which are communication closed 
	in the sense that no message received in round $t$ is sent in a round different from~$t$\footnote{%
	Synchronized rounds can be implemented in totally asynchronous systems, i.e.,  when there is 
	no bound on message delays and relative process speed~\cite{Awe85}, and even in the presence of 
	non-malicious failures~\cite{CBS09}.}.
In each round, an active node successively sends messages, receives some messages, and 
	then undergoes an internal transition to a new state.
The period from round $t$ to round~$t'$, $t'\geq t$, is denoted $[t,t']$.
Communications that occur at round $t$, i.e., which nodes receive messages from which nodes at round~$t$, 
	correspond to the directed graph 
	$\dG(t)=([ n ] ,E_t)$ where $[ n ]  = \{0, \dots, n-1\}$ and $(i,j)\in E_t$ if and only if the agent~$j$
	 receives the message sent by $i$  to~$j$ at round~$t$.
A passive node at round~$t$ thus corresponds to an isolated node in $\dG(t)$, with an in-degree and 
	an out-degree both equal to zero.
We assume a self-loop at each active node in~$\dG(t)$, since an agent can communicate with itself instantaneously.
The \emph{network} is thus modeled by the \emph{dynamic graph}~$\dG$,
	i.e., the infinite sequence of directed graphs~$\dG=\left (\dG(t) \right )_{t \geq 1}$,
	with the same set of nodes.
Observe that each node~$i$ has an in-degree and an out-degree in $\dG$ both equal to zero up to 
	a certain index~$t_i$, and then has a self-loop beyond that  index~$t_i$, which represents the 
	$i$'s activation time.
	
Each node may possess more or less information about the network it belongs to.
 The knowledge of certain information about  the network corresponds to 
 	restrict to a non-empty subset of networks, called a \emph{network class}.
As an example, the network class where the number of nodes is known to be $n$
	is captured by the set of dynamic graphs with $n$ nodes.
Similarly,  the network class where an upper bound $N$ on the number of nodes is known 
	is the set of dynamic graphs with at most $N$ nodes.
The class of \emph{symmetric networks} corresponds to the set of dynamic graphs 
	with bidirectional edges -- that is, at each round~$t$, $(i,j) \in E_t$ if and only if $(j,i) \in E_t$.
		
An \emph{algorithm}~\textsc{alg} is given by a non-empty set $\mathcal{Q}$ of local states,
	a non-empty subset $\mathcal{Q}_0\subseteq \mathcal{Q}$ of initial states,
	a non-empty set of messages~$\mathcal{M}$, a sending function, and a transition function.
The latter function determines the state after a transition by an active node:
	the new state is computed on the basis of the current state
	and the messages that have been just received.
That corresponds to a transition function~$\tau : \mathcal{Q} \times \mathcal{M}^{\oplus} \rightarrow \mathcal{Q}$,
	where $\mathcal{M}^{\oplus}$ denotes the set of non-empty finite multi-sets over the set~$\mathcal{M}$.
The messages to be sent by an active  node just depend on its current state,
	and thus correspond to  a sending function $\sigma : \mathcal{Q} \to \mathcal{M}$.

An \emph{execution of} the algorithm~\textsc{alg} in the dynamic graph~$\dG$
	proceeds as follows:
At the beginning of round $t_i$, the state of the node~$i$ belongs to the set of initial states~$\mathcal{Q}_0$.
In each round~$t = t_i, t_i +1,\dots$, the node~$i$ applies the sending function~$\sigma$
	to generate the message to be sent, then it receives the messages sent by
	its incoming neighbors in the directed graph~$\dG(t)$, and finally applies the transition function $\tau$
	to its current state and the multi-set of messages it has just received to go to a next state.
Algorithms are deterministic, and thus any execution of an algorithm is defined by the initial state of the network
	and  the dynamic graph $\dG$.	
In the rest of the paper, we adopt the following notation: given an execution of~\textsc{alg}, the state of 
	a node~$i$ at the end of round~$t$ is denoted $s_i(t)$, and $s_i(0)$ is its initial state .
Similarly, the value at the end of round~$t$ of any variable $x_i$, local to the node~$i$, is denoted $x_i(t)$, 
	and $x_i(0)$ is the initial value of $x_i$.	
It follows that  if $t_i$ denotes the round in which the node $i$ becomes active, then $s_i(t)=s_i(0)$ for every 
	$t \in \{0,\cdots,  t_i-1\}$.

\subsection{The synchronization problem}
Let  ${\cal C}$ be a network class, and let  \textsc{alg} be an algorithm where 
	each node~$i$ maintains a local counter, called \emph{local clock},~$C_i$.
We fix an integer $P>1$ and say that the algorithm \textsc{alg} \emph{solves the mod~$P$-synchronization 
	problem in the network class} ${\cal C}$ if, for every execution of \textsc{alg} in a network in 
	${\cal C}$, the corresponding sequences of integers $(C_i(t))_{i \in [n], t \in \mathbb{N}}$ satisfy
\begin{equation*}
    \exists c \in \mathbb{Z}, 
    \exists t_0 \in \mathbb{N}, \forall t \geq t_0, \forall i \in [n],~C_i(t) \equiv_P t + c.
\end{equation*}
In other words, from a certain round, all clocks are congruent modulo $P$ and 
	are incremented by one modulo $P$ in each round.
Observe that the nodes are only required to synchronize their clocks, without awareness that
	synchronization is reached (no \emph{firing}).
The issue addressed in this paper is the existence of a synchronization algorithm using bounded memory,
	that is,  with a finite set ${\cal Q}$ of states and  working in network classes where 
	 the number of nodes is arbitrary.

\subsection{Asynchronous starts and self-stabilization}
	
Much of the literature on asynchronous starts considers the more restrictive model of \emph{diffusive starts}:
	a passive node is able to hear of an incoming neighbor, and becomes active upon the receipt of a message\footnote{%
The model of diffusive starts can be interestingly compared with the model with 
	\emph{heartbeats}~\cite{BurL:acr:1987,CBM:tcs:19}
	where, on the contrary, passive nodes can receive nothing, but send null messages.}.
In other words, an active node wakes up all its outgoing neighbors.

Actually, the model of asynchronous starts corresponds to arbitrary initializations of the local clocks.
The synchronization problem can be extended to the model where the whole local state of each
	agent is arbitrary.
In this case, a synchronization algorithm  is said to be  \emph{self-stabilizing}.
For such algorithms, there is clearly no restriction in assuming  that nodes all start
	in the first round: every algorithm that synchronizes clocks despite arbitrary 
	initializations of the nodes when they are all active in the first round still works with asynchronous 
	starts, whether diffusive or not.
We thus obtain a hierarchy of three models for the synchronization problem, namely the model of diffusive starts, 
	the model of asynchronous starts, and the self-stabilization model (with synchronous starts).
The point of  this paper is to study whether synchronization is possible with bounded memory in 
	each of  these models.

\section{Self-stabilizing synchronization}

%\subsection{Self-stabilization}
%
%The point of the notion of self-stabilization is to tolerate \emph{transient faults}.
%A \emph{transient fault} is an arbitrary modification of the state of a system at a certain instant.
%As an example, a possible source of transient faults in real-world systems is cosmic rays, that can randomly flip a bit in the memory of a digital device.
%A self-stabilizing algorithm is an algorithm that is able to recover from any transient fault in finite time.
%In other words, in any execution,
%    a certain correctness property is eventually achieved regardless of the initial state of the nodes.
%More specifically, an algorithm \textsc{alg} is said to be \emph{self-stabilizing} if its set of initial states $\mathcal{I}$ is equal to its set of states $\mathcal{Q}$.

\subsection{Our result}
	
In this section, we prove that there is no self-stabilizing synchronization algorithm with bounded memory 
	for  the class  of networks with a finite \emph{dynamic diameter} (see Section~\ref{sec:asynch_starts}).
Actually, our proof only involves static unidirectional rings, which shows  the impossibility result  holds for the 
	restricted class $\Cst$ of static networks that are strongly connected, and even holds
	 in the subclasses of networks  with bounded  in- or out-degrees.
Moreover, a refinement of the argument works with bidirectional rings, yielding the same impossibility result 
	when limited to the subclass of networks $\Cstbi$  with symmetric communication links.
For the sake of clarity, we start with the proof for $\Cst$, and then briefly adapt the argument for $\Cstbi$.

\begin{thm} \label{thm:no_self_stab}
For any period $P > 1$,
    no self-stabilizing algorithm solves the mod~$P$-synchronization problem using bounded memory 
    in the class of static networks that are strongly connected.
\end{thm}
\begin{proof}
The proof is by contradiction: assume there is a self-stabilizing algorithm \textsc{alg}, with a finite set of states 
	${\cal Q}$,  that achieves  mod~$P$-synchronization  in any strongly connected static network.
For an arbitrary pair of states $q^0, q^1 \in \mathcal{Q}$, we inductively define the
	sequence of states $(q^r)_{r \in \mathbb{N}}$ as:
	\begin{equation}\label{eq:defrecq}
        q^{r+2}  \eqdef \tau \left ( q^{r+1},   \mbset   \sigma(q^{r+1}), \sigma(q^{r})  \meset  \right ).
	\end{equation}
The set ${\cal Q}^2$ is finite, since ${\cal Q}$ is  finite, and  this second-order sequence is 
	ultimately periodic (cf. Figure~\ref{fig:mem}):
	there exist two positive integers $\ell $ and $L$ such that 
	$$ \forall r \geq \ell -1, \ \ q^{r+L} = q^r .$$

Let us consider the execution of \textsc{alg} in the directed ring of size $L$
	with a self-loop at each node, and where every node $i \in [L]$ starts in the state~$s_i(0) \eqdef q^{i+\ell}$
	(cf. Figure~\ref{fig:memtime}).
We then prove, by induction on $t \in \mathbb{N}$,  that 
	\begin{equation} \label{eq:state_sys}
	\forall i \in [L],\ \ s_i(t) = q^{i+\ell+t}.
	\end{equation}
\begin{enumerate}
	\item The base case $t=0$ follows from the definition of the initial states.
        \item Assume that Eq.~(\ref{eq:state_sys}) holds for a certain integer $t\in  \mathbb{N}$. 
        In round $t+1$, each node $i$ receives a message from itself and from its predecessor in the ring, 
		namely  $L-1$ for the node $0$ and $i-1$ otherwise.
	It follows that for every node $i\neq 0$,
		$$  s_i(t+1)  = \tau \left ( q^{i+\ell+t},   \mbset   \sigma(q^{i+\ell+t}), \sigma(q^{i-1+\ell+t}) \meset  \right ) $$
		and
		$$ s_0(t+1) = \tau \left ( q^{\ell+t},   \mbset  \sigma(q^{\ell+t}), \sigma(q^{L-1+\ell+t}) \meset  \right ) 
                	= \tau \left ( q^{\ell+t},   \mbset   \sigma(q^{\ell+t}), \sigma(q^{\ell-1+t }) \meset  \right ) $$
	since $t$ is non-negative and  $L$ is a period of the sequence $(q^r)_{r \geq \ell -1}$.
	In both cases, the inductive definition of the sequence $(q^r)_{r \in \mathbb{N}}$ in Eq.~(\ref{eq:defrecq}) yields
         $ s_{i}(t+1)  = q^{\ell+i+t+1}  $, as required.
\end{enumerate}
It follows that in each round $t$,  we have
    \begin{equation}\label{eq:ringuni}
       s_1(t) = q^{\ell+t+1} = s_0(t+1).
    \end{equation}
Since \textsc{alg} achieves mod~$P$-synchronization, there exists some round~$t_s$ such that 
	for all rounds $t \geq t_s$, it holds that
	\begin{equation*}
     C_1( t ) \equiv_P C_0(t) \  \mbox{ and } \ C_0(t+1)  \equiv_P  C_0(t) + 1,
        \end{equation*}
	which, with Eq.~(\ref{eq:ringuni}), yields
	$$  C_0(t) + 1   \equiv_P C_0(t) ,$$
	a contradiction.
\end{proof}
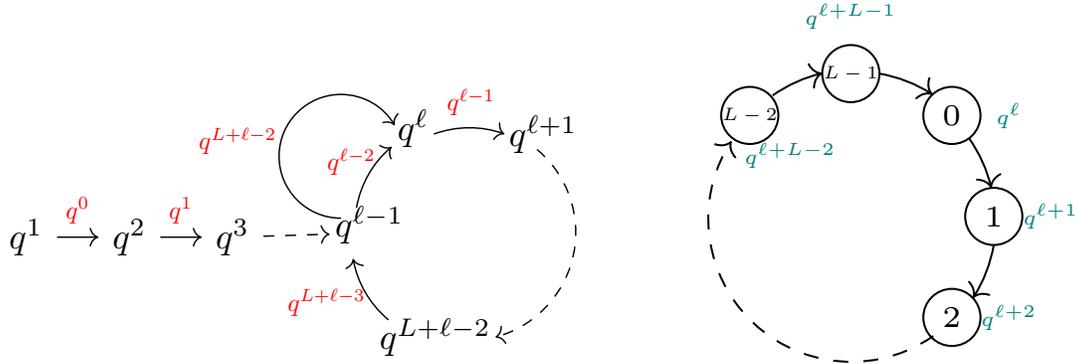
\begin{figure}
\centering
\begin{subfigure}[b]{0.5\textwidth}
\centering
\resizebox{\linewidth}{!}{
\begin{tikzpicture}[font=\small]
    \node (n1) {$q^1$};
    \node (n2) [right of=n1] {$q^2$};
    \node (n3) [right of=n2] {$q^3$};
    \draw[->] (n1) -- node[midway, above]{\scalebox{0.7}{$\color{red} q^0$}} (n2);
    \draw[->] (n2) -- node[midway, above]{\scalebox{0.7}{$\color{red} q^1$}} (n3);

    \begin{scope}[shift={(3.6,0.9)}]
        \node (n4) [label={[shift={(0.15, -0.3)}]$q^\ell$}]{};
        \node (n5) [rotate=60, left of=n4, label={[shift={(0.81, -0.3)}]$q^{\ell-1}$}] {};
        \node (n6) [rotate=120, left of=n5, label={[shift={(1.15, 0.25)}]$q^{L+\ell-2}$}] {};
        \node (n7) [rotate=180, left of=n4, label={[shift={(0.41, 0.55)}]$q^{\ell+1}$}] {};
        \draw[->] (n5) ++(0.1, 0.25) arc (170:130:1);
        \draw[->] (n5) ++(-0.05, 0.15) arc (270:30:0.6);
        \draw[->] (n6) ++(-0.09, 0.03) arc (230:190:1);
        \draw[->] (n4) ++(0.35, 0.1) arc (110:70:1);
        \node (nl2) [above of=n5,
            label={[shift={(0.05, -0.7)}]\scalebox{0.7}{\color{red} $q^{\ell-2}$}}] {};
        \node (nl3) [above of=n5,
            label={[shift={(-1.05, -0.5)}]\scalebox{0.7}{\color{red} $q^{L+\ell-2}$}}] {};
        \node (nl4) [above of=n7,
            label={[shift={(-0.29, -0.99)}]\scalebox{0.7}{\color{red} $q^{\ell-1}$}}] {};
        \node (nl5) [above of=n6,
            label={[shift={(-0.69, -1.20)}]\scalebox{0.7}{\color{red} $q^{L+\ell-3}$}}] {};
    \end{scope}

    \draw[dashed, ->] (n3) -- (n5);
    \draw[dashed, ->] (n7) ++(0.35, -0.05) arc (50:-75:1.05);
\end{tikzpicture}
}
\caption{The arrow $q^r \stackrel{\textcolor{red}{q^{r-1}}}{\longrightarrow} q^{r+1}$ represents 
	the state transition $q^{r+1}  = \tau ( q^{r},   \mbset   \sigma(q^{r}), \sigma(q^{r-1})  \meset  )$}
\label{fig:mem}
\end{subfigure}
\hspace{2em}
\begin{subfigure}[b]{0.35\textwidth}
\centering
\resizebox{\linewidth}{!}{

\begin{tikzpicture}[node distance=0.4cm]
    \node[circle,draw,minimum size=0.4cm,inner sep=-0.1] (n1) at (135:1) {\scalebox{0.6}{\tiny $L-2$}};
    \node[circle,draw,minimum size=0.4cm,inner sep=-0.1] (n2) at (90:1) {\scalebox{0.6}{\tiny $L-1$}};
    \node[circle,draw,minimum size=0.4cm,inner sep=-0.1] (n3) at (45:1) {\tiny $0$};
    \node[circle,draw,minimum size=0.4cm,inner sep=-0.1] (n4) at (0:1) {\tiny $1$};
    \node[circle,draw,minimum size=0.4cm,inner sep=-0.1] (n5) at (-45:1) {\tiny $2$};

    \node[below right of=n1] (qm1) {\scalebox{0.7}{\tiny \color{teal} $q^{\ell+L-2}$}};
    \node[above of=n2] (qm2) {\scalebox{0.7}{\tiny \color{teal} $q^{\ell+L-1}$}};
    \node[right of=n3] (qm3) {\scalebox{0.7}{\tiny \color{teal} $q^{\ell}$}};
    \node[right of=n4] (qm4) {\scalebox{0.7}{\tiny \color{teal} $q^{\ell+1}$}};
    \node[right of=n5] (qm5) {\scalebox{0.7}{\tiny \color{teal} $q^{\ell+2}$}};

    \draw[->] (n1) ++(0.16,0.14) arc (125:103:1);
    \draw[->] (n2) ++(0.20,0.00) arc (80:57:1);
    \draw[->] (n3) ++(0.12,-0.16) arc (35:13:1);
    \draw[->] (n4) ++(0,-0.2) arc (-10:-33:1);
    \draw[dashed, ->] (n5) ++(-0.13,-0.13) arc (-55:-213:1);
\end{tikzpicture}
}
\caption{Oriented ring of size $L$ and initial states in green.
}
\label{fig:memtime}
\end{subfigure}
\caption{State graph and initial state for the directed ring of size $L$.}\label{fig:twotrees}
\end{figure}

\subsection{Impossibility with bidirectional links}

We now show that the above impossibility result still holds when we restrict ourselves to the class 
	of symmetric networks~$\Cstbi$.

\begin{thm} \label{thm:lowerbounds}
    For any period $P > 1$, 
    no self-stabilizing and bounded-memory algorithm solves the mod~$P$-synchronization problem 
    in the class of static strongly connected networks with bidirectional links.
\end{thm}

\begin{proof}
The argument is by contradiction: assume there is a self-stabilizing algorithm~\textsc{alg} with a finite
	set of states ${\cal Q}$  that achieves mod $P$-synchronization in  the class of static 
	strongly connected networks with bidirectional links~$\Cstbi$.
We fix three states $p^0, q^0, q^1\in \mathcal{Q}$, and we inductively define the two sequences 
	$(p^r)_{r \in \mathbb{N}}$ and $(q^r)_{r \in \mathbb{N}}$ by:
	$$ p^{r+1}  \eqdef  \tau \left (\,  p^{r},   \mbset   \sigma(q^{r}), \sigma(p^{r}) , \sigma(q^{r + 1})  \meset  \, \right ) \ \mbox{ and } \ 
		q^{r+1}  \eqdef  \tau \left ( \, q^{r},   \mbset   \sigma(p^{r-1}), \sigma(q^{r}) , \sigma(p^{r })  \meset  \, \right ) . $$

Since $\mathcal{Q} ^2 $ is finite, the sequence $(p^r, q^r)_{r \in \mathbb{N}}$ is ultimately periodic.
Let $\ell$ be a positive integer such that this sequence is periodic from round $\ell - 1$ and let $L$ be its period.

Then we construct the execution of \textsc{alg} in the bidirectional ring of size $2L$
	where nodes start in the following states:
	\begin{equation*}
        s_i(0) \eqdef \left \{ 
        \begin{array}{ll}
        q^{\ell+k } & \mbox{if  } i= 2k \\
        p^{\ell+k} & \mbox{if }  i= 2k +1.
        \end{array} \right. 
         \end{equation*}
A similar inductive argument as above shows that the state of each node $i \in \{0, 1, \cdots, 2L-1 \}$ at the end of 
	round $t $ is equal to
	\begin{equation*}
        s_i(t)  =  \left \{ 
        \begin{array}{ll}
        q^{\ell+k +t } & \mbox{if  } i= 2k \\
        p^{\ell+k +t } & \mbox{if }  i= 2k +1.
        \end{array} \right. 
         \end{equation*}
  In particular, we have $s_2( t ) = s_0(t+1)$.
      
 Since \textsc{alg} achieves mod~$P$-synchronization in the bidirectional ring of size $2L$,
 	there exists some round~$t_s$ such that for all rounds $t \geq t_s$, we have
	\begin{equation*}
      C_2( t ) \equiv_P C_0(t) \mbox{ and } C_0(t+1) \equiv_P C_0(t) + 1,
        \end{equation*}
	which, with  $s_2( t ) = s_0(t+1)$, yields
	$ C_0(t) + 1   \equiv_P C_0(t) $,
	a contradiction.
\end{proof}

\subsection{Some lower bounds for self-stabilizing synchronization}\label{sec:lowerboundsSS}

Theorem~\ref{thm:no_self_stab} proves the claim in the title for self-stabilizing solutions but,
	from the class of executions in directed rings considered in its  proof,  we can obtain lower bounds
	on the number of states and on the runtime of the self-stabilizing algorithms that  achieve synchronization 
	in networks with at most $n$ nodes, yielding the following refinement of its statement.
	
\begin{thm}\label{thm:lowerboundSS}
Any self-stabilizing algorithm that achieves synchronization in the class ${\cal C}^*_{\leq n}$
	of static strongly connected networks with at most $n$ nodes requires each node to have 
	 at least equal to $n+1$ states, and  its synchronization 
	time is at least of $n-2$ rounds, in the worst case.
\end{thm}

\begin{proof}
Let $n$ be a positive integer~$n$, and assume that  an algorithm~\textsc{alg},
	with a set of states ${\cal Q} $, achieves mod $P$-synchronization  in a self-stabilizing way 
	in the network class ${\cal C}^*_{\leq n}$.
As  in the proof of Theorem~\ref{thm:no_self_stab}, we fix a pair of states $q^0, q^1$ and define the  sequence
	$(q^r)_{r \in \mathbb{N}}$, which is ultimately periodic.
Let~$\ell-1$ be the smallest index from which the sequence is $L$-periodic. 

For the first claim, assume for contradiction that the cardinality of $ {\cal Q} $ is fewer or equal $n$.
Then, $L \leq n$ and the directed ring of size $L$ is in the network class ${\cal C}^*_{\leq n}$.
We have proved above that if the successive nodes  in this ring start with the initial states
	$q^{\ell}, \cdots, q^{\ell +L-1}$, then they never synchronize, yielding a contradiction.
Hence, we have $|{\cal Q}| \geq n +1$.

For the second claim, let us consider the same type of execution of~\textsc{alg}, but in the directed ring of size $n$,
	where every node $i \in [n] $ starts in the state $q^{i+\ell}$.
By an easy induction, we prove that for every $t\in \{1,\cdots,n-1\}$, this execution is indistinguishable from that in 
	the ring of size $L$, at least from the viewpoint of each of the nodes $t, t+1, \cdots, n-1$.
We thus obtain
	$$ \forall t \in \{1,\cdots,n-1\}, \ s_t(t) = q^{\ell + 2t}, \ s_{t+1} (t) = q^{\ell + 2t+1},\  \cdots, \ s_{n-1} (t) = q^{\ell + t + n-1} .$$
In particular, we have
	$$\left \{ \begin{array}{lll}
	s_{n-2}(n-3) = q^{\ell + 2n-5}  & \mbox{and} & \ s_{n-1}(n-3) = q^{\ell + 2n- 4} \\
	s_{n-2}(n-2) = q^{\ell + 2n-4}  &  \mbox{and} & \   s_{n-1}(n-2) = q^{\ell + 2n- 3} ,
	\end{array}\right.$$	
	which implies that 
	\begin{equation}\label{eq:ring_n}	
	 C_{n-1}(n-3) \equiv_{P}  C_{n-2}(n-2) .
	 \end{equation}
If synchronization is achieved by round $n-3$, then it holds that $ C_{n-1}(n-3) \equiv_{P}  C_{n-2}(n-3) $,
	$ C_{n-1}(n-2) \equiv_{P}  C_{n-2}(n-2) $, and $ C_{n-1}(n-2) \equiv_{P}  C_{n-1}(n-3) +1 $,
	which contradicts Eq.~(\ref{eq:ring_n}) and completes the proof.
\end{proof}

These lower bounds can be interestingly compared with  the space and time complexities of the self-stabilizing
	algorithms proposed in the literature, which use \emph{unbounded but  finite} memory.
In the case where no bound on the network size is available, only one self-stabilizing algorithm to our knowledge, 
	namely the SAP algorithm~\cite{CBPM:opodis:22} with self-adaptive periods, has been proposed.
This algorithm is parametrized with a non-decreasing and inflationary function $g : \IN \rightarrow \IN$, which can be
	tuned to favor either time or space complexity: the faster $g$ grows, the lower the synchronization time is, 
	and the higher its space complexity is.
Table~\ref{tab:self_stab}  provides the synchronization time and the number of bits used by each node in the SAP 
	algorithm in two typical cases, namely  $g = \lambda x.x+1$ and  $g = \lambda x.2x$. 
The results in terms of memory usage do not take into account the size of the initial values, since the latter is arbitrary
	for self-stabilizing algorithms without any additional information (typically, a bound on the network size).
Actually, the SAP algorithm  also works in the class ${\cal C}$ of dynamic networks with a finite 
	\emph{dynamic diameter} 	(see~Section~\ref{sec:asynch_starts} for a definition).
Its time and space complexity are then identical to those obtained in the case of a static strongly connected 
	network, when the dynamic diameter is substituted for the network size.	
	
When a bound $N$ on the network size  is given, i.e., in the class  $\Cst_{\,\leq N}$ of
	strongly connected static networks with at most $N$ nodes,  a judicious choice of the function~$g$, namely 
	$g = \lambda x. \left \lceil \frac{2N}{P}\right\rceil$, provides a linear synchronization time
	for the SAP algorithm, which does not depend on the bound~$N$.
Moreover, the number of states used by each node is $ P \left\lceil 2N/P\right \rceil  $. 
This demonstrates that the runtime bound in Theorem~\ref{thm:lowerboundSS} is asymptotically tight.
The asymptotic tightness also holds for space complexity when  the upper bound $N$ on the number of 
	nodes $n$  is within a constant multiplicative factor of $n$.
In the case the period $P$ is greater than~4, another self-stabilizing algorithm for the mod-$P$ synchronization
	problem has been proposed 	by Feldmann et al.~\cite{FelKS:spaa:20} for  the restricted class of 
	static symmetric networks~$ \Cbi^* \cap {\mathcal{C}}^*_{\,\leq N} $.
Its runtime is also linear in the number of nodes and its space complexity is of the same order than 
	SAP$_{\lambda x. \left \lceil \frac{2N}{P}\right\rceil}$.
This algorithm is remarkable in that it works in the \emph{Beeping model}, but its correctness highly 
	relies on the assumptions of diffusive starts and symmetric communication links.
	
\begin{table}[h]
\begin{tabular}{|p{3cm}|p{3.4cm}|p{2.7cm}|p{2cm}|}
	\hline
    & & & \\[-0.3cm]
                     &  synchronization time &  state number  & validity \\[0.14cm]
	\hline
    & & & \\[-0.3cm]
	
	SAP$_{\lambda x.x + 1}$      &   $ O(n^2) $ & $ O( n^4) $ & $ \mathcal{C}^* $ \\[0.14cm]
	\hline
    & & & \\[-0.3cm]
	SAP$_{\lambda x. 2x }$ & $ O(n\log n) $ &  $O(n!) $  &  $ \mathcal{C}^* $  \\[0.14cm]
	\hline 
    & & & \\[-0.3cm]
	Algorithm 3 in~\cite{FelKS:spaa:20} & $5 (n-1)$ & $   50 P^2 N $ & $ \Cbi^* \cap {\mathcal{C}}^*_{\,\leq N} $ \\[0.14cm]
	\hline
    & & & \\[-0.3cm]
	SAP$_{ \lambda x. \left \lceil \frac{2N}{P} \right \rceil}$  & $3 (n-1) $ & $ P \left\lceil 2N/P\right \rceil $  & ${\mathcal{C}}^*_{\,\leq N}$  \\[0.14cm]
	\hline
\end{tabular}\vspace{0.2cm}
	\caption{Complexity bounds of self-stabilizing synchronization algorithms 
	in static networks.}
\label{tab:self_stab}
\end{table}

\section{Clock synchronization in  dynamic networks }\label{sec:asynch_starts}

In this section, we demonstrate that synchronization is impossible with bounded memory in the 
	dynamic setting.
	
We first recall the definition of \emph{dynamic diameter} which extends that of diameter of a directed 
	graph.
For that, we introduce the \emph{product} of two directed graphs $G_1 = ([n], E_1)$ and  $G_2 = ([n], E_2)$, 
	denoted by $G_1 \circ G_2$, that is a directed graph defined as follows:
	$G_1 \circ G_2$ has $[n]$ as its set of vertices, and $(i,j)$ is an arc if there exists $k \in [n]$ such that 
	$(i,k) \in E_1$ and $(k,j) \in E_2$.
For any dynamic graph $\mathbb{G}$ and any integer $t' > t \geq 1$, we let 
\begin{equation*}
	\mathbb{G}(t:t') \eqdef \mathbb{G}(t) \circ \mathbb{G}(t+1) \circ \dots \circ \mathbb{G}(t').
\end{equation*}
By extension, we let $\mathbb{G}(t:t) = \dG(t)$.
The \emph{dynamic diameter} of~$\dG$ is then defined as: 
\begin{equation*}
	 D(\dG) \eqdef \inf~\{d > 0 \, \mid \, \forall i,j \in [n], \forall t \in \mathbb{N},~(i,j)~\text{is an arc of}	
	 ~\mathbb{G}(t+1:t+d)\}.
\end{equation*}
By defining the notion of \emph{temporal path} between two nodes $i$ and $j$ as a sequence of 
	vertices $i = k_0, k_1, \dots k_\ell = j$ such that, for some non-negative integer $t$, the arcs 
	$(k_0,k_1), (k_1,k_2), \dots, (k_{\ell-1},k_\ell)$ belong to the directed graphs 
	$\mathbb{G}(t+1), \dots, \mathbb{G}(t+\ell)$, respectively, the dynamic diameter  is then the length 
	of the longest shortest temporal path between any pair of nodes.
In the static case, where all the directed graphs $\dG(t)$ are equal to some directed graph $G$, the 
	dynamic diameter of $\dG$ coincides with $G$'s diameter.
	
\subsection{Impossibility result }

\begin{thm} \label{thm:no_dynamic_algo}
For any period $P > 1$, no algorithm solves the mod~$P$-synchronization problem using bounded memory 
	in the class of networks with a finite dynamic diameter, and even in the case of diffusive starts.
\end{thm}
\begin{proof}
The argument is by contradiction: assume there is an algorithm
	\textsc{alg}$= (\mathcal{Q}, \mathcal{Q}_0, \mathcal{M}, \sigma, \tau )$ with a finite
	set of states ${\cal Q}$  that achieves mod P-synchronization in  the class of networks 
	with a finite dynamic diameter and that  tolerates diffusive start schedule.
	
We start the proof by  introducing some notation: 
	for every pair $(a,b)\in \IN\times\IN_{>0}$, let $\rem{a}{b}$ denote the remainder of the euclidean
	 division of $a$ by $b$, i.e.,  $\rem{a}{b} \eqdef a -   \lfloor a /b \rfloor b$.
Moreover, $a^!$ will stand for the least common multiple of all the positive integers less or
	equal to $a$:
	\begin{equation*}
	a^! \eqdef \lcm (1, 2, \cdots, a) .
	\end{equation*}

First we give  a simple general lemma on sequences defined by a first order recurrence in a finite set.
\begin{lem}\label{lem:ultimperiod}
Every  sequence defined by a first order recurrence  in a finite set $X$ 
	is $|X |^{!}\!$-periodic  from the index $|X |^{!}\! -1$.
\end{lem}
\begin{proof}
Let $(x_t)_{t\in\IN}$  be a sequence defined by a first order recurrence  in $X$. 
Since  $X$ is finite, the sequence is ultimately periodic.
That is to say,  there exist two positive integers $k_0$ and $K$ both less or equal to $|X|$ such that 
	\begin{equation*}\label{eq:ultimperiod}
	\forall t \geq k_0 -1, \ \ x_{t+K} = x_t .
	\end{equation*}
Observe that the above property with the integers $k_0$ and $K$ also  holds for any integer 
	$k\geq k_0$ and for any multiple $pK$ of $K$.
By definition,  $ |X |^{!} \geq |X| $ and $|X |^{!}$ is a multiple of $K$, and the lemma follows.	
\end{proof} 

 We now consider a networked system with  the set of nodes $0,\cdots,2L-1$, where 
 	$L= |{\cal Q}|^!$, and fix an initial state $q^0_0 \in \mathcal{Q}_0$.
We construct an execution of 
	\textsc{alg} in this system as follows.
The initial state of all nodes is $q_0^0$, and the edges of  the communication graph at round~$t$ 
	are all the edges $(i,j)\in [2L]^2$  such that 
	$$t \geq 2 + \max ( \rem{i}{L} ,\rem{j}{L})$$ 
	and one of the following condition is satisfied:
	\begin{itemize}
	\item $i=j$; or
	\item  $ \lfloor (t - 1)/L \rfloor \equiv_4 0  $ and $ i \leq L -1$ and $ j = L + \rem{t -1}{L} $; or
	\item $ \lfloor (t -1) /L \rfloor \equiv_4 2 $  and $  i \geq L $ and $j =\rem{t-1}{L}$.
	\end{itemize}
That defines a dynamic graph~$\dG$ which corresponds to the start schedule where each node~$i $ starts in
	round $ t_i =\rem{i}{L} + 2 $, and thus follows the model of diffusive starts.
Moreover,  all the nodes are isolated, i.e., receive no external message, in each period
	from round~$(2k-1)L +1$ to round~$2kL$.
For the rest of the time, nodes are gathered into two groups, namely the groups $\{0, \cdots, L-1\}$
	and $\{L, \cdots, 2 L-1\}$.
In the period $[4kL+1,(4k+1)L]$, alternatively, each node of the second group receives messages from 
	all the nodes in the first group while, in the period $[(4k +2) L+1,(4k+3)L]$, alternatively, each node 
	of the first group receives messages from all the nodes in the second group. 
Clearly, the dynamic graph~$\dG$ is  $4L$-periodic.
Its dynamic diameter is finite and less than or equal to $6L$ (see Figure~\ref{fig:bipartite}).
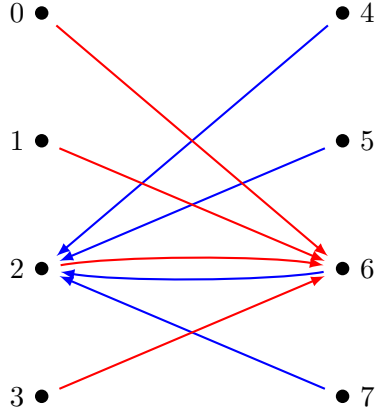
\begin{figure}
\centering
\begin{tikzpicture}
    % Set the style for the nodes
    \tikzset{mynode/.style={circle, fill=black, inner sep=-1.8}}

    % Draw the nodes and labels for the first group
    \foreach \i in {0,1,2,3} {
        \node[circle, fill=white, inner sep=-5] (A\i) at (0, -1.7 * \i) {};
        \node[mynode, label=left:\i] at (0, -1.7 * \i) {};
    }

    % Draw the nodes and labels for the second group
    \foreach \i in {4,5,6,7} {
        \node[circle, fill=white, inner sep=-5] (B\i) at (4, - 1.7 * \i + 1.7 * 4) {};
        \node[mynode, label=right:\i] at (4, - 1.7 * \i + 1.7 * 4) {};
    }

    % Draw the edges between the two groups
    \foreach \i in {2} {
        \foreach \j in {4,5,7} {
            \draw[latex-,thick,blue] (A\i) -- (B\j);
        }
        
    }
    \foreach \i in {0,1,3} {
        \foreach \j in {6} {
            \draw[-latex,thick,red] (A\i) -- (B\j);
        }
    }
	\draw[-latex,thick,red] (A2) .. controls +(10:10mm) and +(170:10mm) .. (B6);
	\draw[-latex,thick,blue] (B6) .. controls +(-170:10mm) and +(-10:10mm) .. (A2);

\end{tikzpicture}
\caption{Red and blue arrows represent the edges in $\dG(4kL+3)$ and $\dG(4kL+2L+3)$ , respectively, 
	when $L=4$.}
\label{fig:bipartite} % FAIRE UNE FIGURE PLUS GENERALE (avec pointilles).
\end{figure}

% correction_louis : j'ai ajouté une phrase pour faire la transition avec ce qui précède
% correction_bernadette : j'ai re-rectifié cette phrase car il y avait un truc qui ne me semblait pas très logique
Given the state $q_0^0 \in \mathcal{Q}_0$ that was fixed earlier, we now define a family of states 
	$(q_k^r)_{k,r \in \mathbb{N}}$ as follows.
A node, starting in the state $q_0^0$ and receiving no external message, updates its state, according 
	to a certain sequence that can be denoted $q_0^1, q_0^2, \dots$.
This sequence is ultimately periodic and is represented by a cherry-shaped subgraph in Figure~\ref{fig:states};
	an isolated node starting in state $q_0^0$ remains trapped in this cycle.
We also define a sequence of states $q_1^0, q_2^0, \dots$ the definition of which is given below.
Each state $q^0_k$  is the starting point of an ultimately periodic sequence, similar to the one starting at $q_0^0$.
By Lemma~\ref{lem:ultimperiod}, the integer~$L$ can be used as the common length of all ``cherry tails'' and 
	``cherry bodies''; this is referred to as  the ``normalized form'' of these sequences in the set~${\cal Q}$.

\begin{figure}
\begin{tikzpicture}[
  vertex/.style = {shape=circle, draw, fill=black, inner sep=-1.9},
  edge/.style = {->, -Latex},
]

% Nodes
\node[vertex,label=left:{$q_0^4$}] (A0) at (0.4, 0) {};
\node[vertex,label=left:{$q_0^1$}] (B0) at (1.5, 1) {};
\node[vertex,label=right:{$q_0^2$}] (C0) at (2.6, 0) {};
\node[vertex,label=below:{$q_0^3$}] (D0) at (1.5, -1) {};
\node[vertex,label=left:{$q_0^0$}] (E0) at (1.5, 2.5) {};

% Edges
\draw[edge] (A0) -- (B0);
\draw[edge] (B0) -- (C0);
\draw[edge] (C0) -- (D0);
\draw[edge] (D0) -- (A0);
\draw[edge] (E0) -- (B0);

% Nodes
\node[vertex,label=below:{$q_1^4$}] (A1) at (4.5, -0.3) {};
\node[vertex,label=left:{$q_1^2$}] (B1) at (5.5, 1) {};
\node[vertex,label=below:{$q_1^3$}] (C1) at (6.5, -0.3) {};
\node[vertex,label=left:{$q_1^1$}] (D1) at (5.5, 2) {};
\node[vertex,label=left:{$q_1^0$}] (E1) at (5.5, 3) {};

% Edges
\draw[edge] (A1) -- (B1);
\draw[edge] (B1) -- (C1);
\draw[edge] (C1) -- (A1);
\draw[edge] (D1) -- (B1);
\draw[edge] (E1) -- (D1);

% Nodes
\node[vertex,label=left:{$q_2^4$}] (A2) at (8.4, 1) {};
\node[vertex,label=left:{$q_2^1$}] (B2) at (9.5, 2) {};
\node[vertex,label=above:{$q_2^2$}] (C2) at (10.6, 1) {};
\node[vertex,label=below:{$q_2^3$}] (D2) at (9.5, 0) {};
\node[vertex,label=left:{$q_2^0$}] (E2) at (9.5, 3.5) {};

% Edges
\draw[edge] (A2) -- (B2);
\draw[edge] (B2) -- (C2);
\draw[edge] (C2) -- (D2);
\draw[edge] (D2) -- (A2);
\draw[edge] (E2) -- (B2);

% Nodes
\node[vertex,label=below:{$q_3^2$}] (A3) at (12.5, -0.3) {};
\node[vertex,label=above:{$q_3^0$}] (B3) at (13.5, 1) {};
\node[vertex,label=below:{$q_3^1$}] (C3) at (14.5, -0.3) {};

% Edges
\draw[edge] (A3) -- (B3);
\draw[edge] (B3) -- (C3);
\draw[edge] (C3) -- (A3);

% \draw[edge,dashed,red] (C0) -- (E1);
% \draw[edge,dashed,red] (B1) -- (E2);
% \draw[edge,dashed,red] (C2) -- (B3);
% \draw[edge,dashed,red] (A3) -- (E1);

\end{tikzpicture} %REFAIRE LA FIGURE avec 5 etats et L=60
\caption{A sequence $\left( q^r_k \right)_{r\in \IN}$ and its normalized form in a space with 5 states.}
\label{fig:states}
\end{figure}
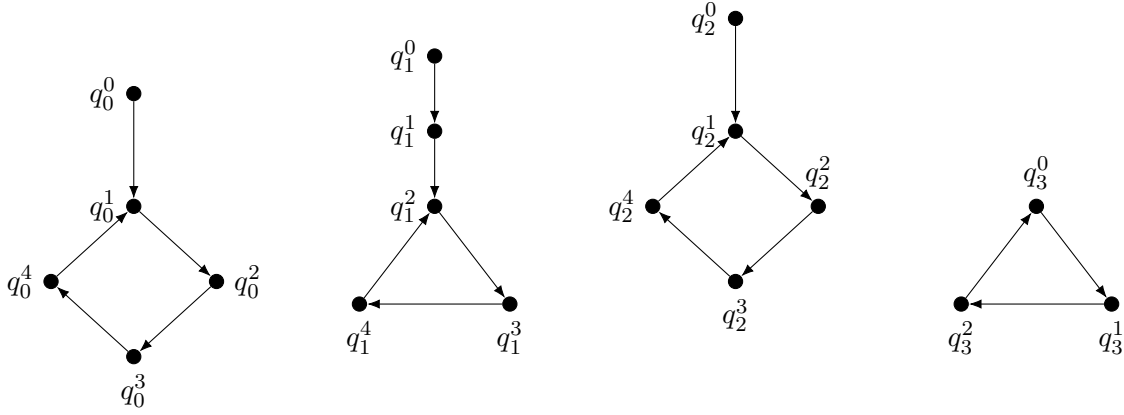	
	
The point of the execution \textsc{alg} with the dynamic graph $\dG$ and the initial states $q_0^0$
	then lies in the following (cf. Figure~\ref{fig:petri}):	
In each round that is a multiple of $2L$, the nodes in one of the two groups are trapped in the loop generated by the 
	state~$q^0_k$, for some integer $k$, while the other nodes are trapped in the loop generated by the 
	state~$q^0_{k+1}$.
During the next $2L$ rounds, the following phenomenon occurs:
One by one, all the nodes that are trapped in the $q^0_k$-generated loop will follow one of the state transitions represented in red in Figure~\ref{fig:petri}.
They will end up trapped in the loop generated by the state $q^0_{k+2}$.
The blue arcs of Figure~\ref{fig:petri} represent the messages that a node in state~$q_k^{L-1}$ must receive 
	to trigger a state transition to $q_{k+2}^0$.
In the meantime, the other nodes remain in the loop generated by $q^0_{k+1}$.

\begin{figure}
\resizebox{\columnwidth}{!}{%
\begin{tikzpicture}[my place/.style={place, minimum size=1.5em},]

% Nodes0
\node[my place, tokens=1] (A0) at (0.2, -0.8) {};
\node[my place, tokens=1, label=left:{$q_\ell^{L-1} = q_\ell^{2L-1}$}] (B0) at (1.5, 1) {};
\node[my place, tokens=1] (C0) at (2.8, -0.8) {};
\node[my place] (D) at (1.5, 2.5) {};
\node[my place,label=left:{$q^0_\ell$}] (E0) at (1.5, 4) {};

% Edges
\draw[-latex,thick] (A0) -- (B0);
\draw[-latex,thick] (B0) -- (C0);
\draw[-latex,thick,dashdotdotted] (C0) -- (A0);
\draw[-latex,dashdotdotted,thick] (D) -- (B0);
\draw[-latex,thick] (E0) -- (D);

% Nodes
\node[my place, tokens=1] (A1) at (5.2, -0.8) {};
\node[my place, tokens=1] (B1) at (6.5, 1) {};
\node[my place, tokens=1] (C1) at (7.8, -0.8) {};
\node[my place] (D) at (6.5, 2.5) {};
\node[my place,label=left:{$q^0_{\ell+1}$}] (E1) at (6.5, 4) {};

% Edges
\draw[-latex,thick] (A1) -- (B1);
\draw[-latex,thick] (B1) -- (C1);
\draw[-latex,thick,dashdotdotted] (C1) -- (A1);
\draw[-latex,thick,dashdotdotted] (D) -- (B1);
\draw[-latex,thick] (E1) -- (D);

% Nodes
\node[my place] (A2) at (10.2, -0.8) {};
\node[my place] (B2) at (11.5, 1) {};
\node[my place] (C2) at (12.8, -0.8) {};
\node[my place] (D) at (11.5, 2.5) {};
\node[my place,label=right:{$q^0_{\ell+2}$}] (E2) at (11.5, 4) {};

% Edges
\draw[-latex,thick] (A2) -- (B2);
\draw[-latex,thick] (B2) -- (C2);
\draw[-latex,thick,dashdotdotted] (C2) -- (A2);
\draw[-latex,thick,dashdotdotted] (D) -- (B2);
\draw[-latex,thick] (E2) -- (D);

% Nodes
\node[my place] (A3) at (15.2, -0.8) {};
\node[my place] (B3) at (16.5, 1) {};
\node[my place] (C3) at (17.8, -0.8) {};
\node[my place] (D) at (16.5, 2.5) {};
\node[my place,label=right:{$q^0_{\ell+3}$}] (E3) at (16.5, 4) {};

% Edges
\draw[-latex,thick] (A3) -- (B3);
\draw[-latex,thick] (B3) -- (C3);
\draw[-latex,thick,dashdotdotted] (C3) -- (A3);
\draw[-latex,thick,dashdotdotted] (D) -- (B3);
\draw[-latex,thick] (E3) -- (D);

\draw[-latex,thick,red] (B0) .. controls +(up:20mm) and +(left:30mm) .. (E2);
\draw[-latex,thick,red] (B1) .. controls +(up:20mm) and +(left:30mm) .. (E3);

\draw[-latex,thick,dashed,blue] (A1) .. controls +(north west:10mm) and +(right:10mm) .. (B0);
\draw[-latex,thick,dashed,blue] (B1) .. controls +(north west:10mm) and +(right:10mm) .. (B0);
\draw[-latex,thick,dashed,blue] (C1) .. controls +(north west:10mm) and +(right:10mm) .. (B0);

\draw[-latex,thick,dashed,blue] (A2) .. controls +(north west:10mm) and +(right:10mm) .. (B1);
\draw[-latex,thick,dashed,blue] (B2) .. controls +(north west:10mm) and +(right:10mm) .. (B1);
\draw[-latex,thick,dashed,blue] (C2) .. controls +(north west:10mm) and +(right:10mm) .. (B1);

\draw[-latex,thick,dashed,blue] (A3) .. controls +(north west:10mm) and +(right:10mm) .. (B2);
\draw[-latex,thick,dashed,blue] (B3) .. controls +(north west:10mm) and +(right:10mm) .. (B2);
\draw[-latex,thick,dashed,blue] (C3) .. controls +(north west:10mm) and +(right:10mm) .. (B2);

\end{tikzpicture}
}
\caption{Global state reached every round that is a multiple of $2L$ round.
    Each vertex represents a state, and each token represents a node.
%    Black arcs represent the application of the transition function, similarly to Figure~\ref{fig:states}.
    Dotted blue arcs represent the messages that the nodes 0 and $L$  receive,
    while red arcs indicate their state transition in the next round.}
%    Some states may be duplicated.}
\label{fig:petri}
\end{figure}
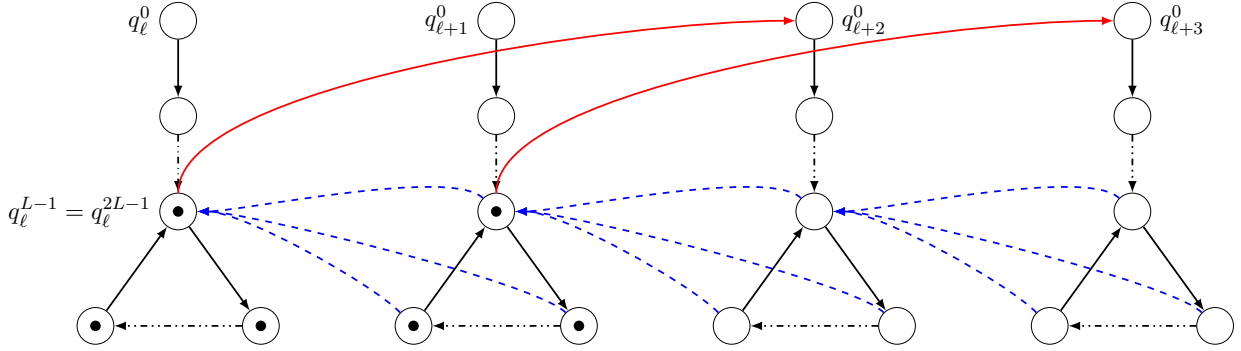

Since \textsc{alg} achieves synchronization modulo $P$, there exists some
	integer $k_0$ such that  the states composing each cherry body 
	$q^{L-1}_k, \cdots, q^{2L-2}_k$ contains all possible output values, from 0 to $P-1$.
Therefore, nodes in the  above defined execution  never synchronize.

We now formalize the above argument.
For that, we prove that the  state of the system in every round of this execution can be computed 
	in terms of the local states $(q_k^r)_{k,r \in \mathbb{N}}$ defined as follows:
 	\begin{equation*}
        q_k^r \eqdef 
        		\begin{cases}
            	q_0^r & \text{if}~k = 1~~~~~~~~~~~~~~~\!~~~~~~~(a) \\
            	\tau \left ( q_{k - 2}^{2L-1}, \bigmbset \sigma(q_{k - 2}^{2L-1}), \sigma(q_{k - 1}^{L}), \sigma(q_{k - 1}^{L+1}),
                	\dots, \sigma(q_{k - 1}^{2L-1}) \meset\right ) & \text{if}~k > 1~\text{and}~r = 0 
	~~~~~~~~(b)\\
            	\tau \left ( q_k^{r-1}, \bigmbset\sigma(q_k^{r-1}) \bigmeset \right ) & \text{if}~r > 0.
	~~~~~~~~~~~~~~~~~~~~~(c)
        		\end{cases}
    \end{equation*}
Observe that every sequence of states $(q_k^r)_{r \in \mathbb{N}}$ is a sequence defined by a first 
	order recurrence, and thus, by Lemma~\ref{lem:ultimperiod} is 
	$|{\cal Q}|^{!}$-periodic  from the index $|{\cal Q}|^{!}\! -1$.

\begin{lem}\label{lem:states}
For every positive integer $k$, the state of the node $i $ at the end of round $2kL$ is equal to    
\begin{equation}\label{eq:localstate}
        s_i(2kL) = q_{k + \lambda}^{2L-1-\rem{i}{L}} ~~~ \text{where} ~~~ \lambda = 
        \begin{cases}
            -1 & \text{if}~i \geq L~\text{xor}~ k \text{ is odd} \\
            0 & \text{otherwise.}
        \end{cases}
    \end{equation}
\end{lem}
\begin{proof}
We proceed by induction on $k\in\IN_{>0}$.
\begin{enumerate}
	\item Base case: $ k = 1$. 
        Every node $i \in [L]$ is in state $q^0_0$ \emph{at the beginning} of round $i+2$ since 
        		$t_i = \rem{i}{L}$.
        Moreover, during the first $2L$ rounds, it receives no message, except its own message
       		from round $t_i$.
	The recurring rule (c) for the definition of the sequence $(q_0^r)_{r \in \mathbb{N}}$ yields
         \begin{equation*}
                s_i(2L) = q_0^{2L-1- \rem{i}{L}} .
         \end{equation*}
            Similarly, for every node $i \in \{L,\cdots, 2L-1\}$, the recurring rules (a) and (c) 
            for defining the sequence $(q_1^r)_{r \in \mathbb{N}}$ lead to
            \begin{equation*}
                s_i(2L) = q_1^{2L-1-\rem{i}{L}}.
            \end{equation*}
            \item Inductive case: we assume that the states of all the nodes at the end of round $2kL$
            	 satisfy Eq.~(\ref{eq:localstate}).
	 By definition of the directed graphs $\dG(2kL)$ and $\dG((2k +1)L)$, the transition from $k$ to $k+1$ 
	 corresponds to a translation of $2L$ rounds, and so consists in switching every node $i \in [L]$ 
	 with the node $i +L$.
	 Without loss of generality, we may thus assume that $k$ is odd.
In this case,  we have 
	\begin{equation*}
	        s_i(2kL) =
        		\begin{cases}
            	q_{k-1}^{2L-1-\rem{i}{L}} & \text{if}~i\in \{0,\cdots,L-1\}   \\
            	 q_{k}^{2L-1-\rem{i}{L} } & \text{if}~i\in \{L,\cdots,2L-1\}  .
	         \end{cases}
    \end{equation*}
In round $2kL+1, 2kL+2, \cdots, 2(k+1)L$, each node $i\in \{L,\cdots, 2L-1\}$  is active since~$k$ is
	positive, and it receives no message except its own message.
Because of the inductive assumption, we have  $ s_i(2kL ) =  q_{k}^{2L- 1 -\rem{i}{L} } $.
By a repeated application of the rule (c), we obtain
	\begin{equation}\label{eq:inducdroite}
	s_i(2kL +1) = q_{k}^{2L-\rem{i}{L} } ,
	\  s_i(2kL +2) = q_{k}^{2L-\rem{i}{L} +1} , \cdots,
	\ s_i(2(k+1)L ) = q_{k}^{4L-1-\rem{i}{L} } .
	\end{equation}	
Since $\rem{i}{L} \leq L-1$, Lemma~\ref{lem:ultimperiod} implies that 
	$ q_{k}^{4L-1-\rem{i}{L} }  = q_{k}^{2L-1-\rem{i}{L} } $.
The last equality  in (\ref{eq:inducdroite}) then proves that Eq.~(\ref{eq:localstate}) holds for  $k+1$ 
	and all the nodes in $\{L,\cdots,2L-1\}$.

In round $2kL+1$, the node 0 receives messages from each of the nodes $L, L+1, \dots, 2L-1$, 
	in addition of its own message.
By the inductive hypothesis, the state of node 0 in round $2kL$ is $q_{k-1}^{2L-1}$, and
	the states of the nodes $L, \dots, 2L-1$ are $q_k^{2L-1}, \dots, q_k^{L}$, respectively.
It follows that
            \begin{equation*}
                s_0(2kL+1) = \tau \left ( q_{k-1}^{2L-1},
                \bigmbset \sigma(q_{k-1}^{2L-1}), \sigma(q_{k}^L), \sigma(q_{k}^{L+1}), \dots, 
                \sigma(q_{k}^{2L-1}) \bigmeset \right )
            \end{equation*}
 	and the recurring rule (b) yields $ s_0(2kL+1) = q_{k+1}^0$.
	
Any node $i \in \{1, \dots, L-1\}$ receives no message in rounds $2kL+1, \cdots, 2kL+ i $.
The inductive hypothesis  and a repeated application of the rule (c) give that its state 
	at the beginning of round $2kL+i+1$ is  $q_{k-1}^{2L-1}$.	
In this round, $i$ receives 	messages from each of the nodes $L, L+1, \dots, 2L-1$, 
	in addition of its own message.
Using Eq.~(\ref{eq:inducdroite}), we obtain
	 \begin{equation*}
                s_i(2kL+i+1) = \tau \left ( q_{k-1}^{2L-1},
                \bigmbset \sigma(q_{k-1}^{2L-1}), \sigma(q_{k}^{2L+i}), \sigma(q_{k}^{2L+i-1}), \dots, 
                \sigma(q_{k}^{L+i+1}) \bigmeset \right ).
            \end{equation*}
Lemma~\ref{lem:ultimperiod} shows that the two multi-sets 
	$\mbset q_{k}^{2L+i} , q_{k}^{2L+i-1}, \dots, q_{k}^{2L}  \meset$ and
	$\mbset  q_{k}^{L}, q_{k}^{L+1}, \dots, q_{k}^{L +i} \meset$  are equal.
Consequently,	$$\forall i \in [L],\ \ s_i(2kL+i+1) = s_0(2kL+1) = q_{k+1}^0 .$$
Later, in all the rounds $2kL+i+2, 2kL+i+3, \cdots, 2(k+1)L$, the node $i$ is again isolated,
	and  a repeated application of the rule (c) yields
	$$s_i(2(k+1)L )= q_{k+1}^{2L-1-i}, $$
	as required.
\end{enumerate}
\end{proof}

As an immediate consequence of Lemma~\ref{lem:states} and its proof, we obtain that, for all positive
	 integers~$k$,  
	 $$ s_0(4kL) = q_{2k}^{2L-1} = s_1(4kL+1), $$
	 and thus 
	$ C_1(4kL+1) \equiv_P  C_0(4kL)    $.
Since \textsc{alg} achieves synchronization modulo $P$, there exists some
	integer $r_0$ such that 
	$$ \forall t \geq r_0, \ \  C_1(t+1) \equiv_P  C_1(t) + 1 \ \mbox{ and } \
			C_1(t) \equiv_P C_0(t)  . $$
Hence,  for $4kL \geq r_0$, we have 
    \begin{equation*}
       C_1(4kL+1) \equiv_P  C_0(4kL) + 1,
    \end{equation*}
    a contradiction with $ C_1(4kL+1) \equiv_P  C_0(4kL)  $.
\end{proof}

Theorem~\ref{thm:no_dynamic_algo} demonstrates that, even in the case of diffusive starts, synchronization is
	impossible with bounded memory in the network class $\Cd$ of dynamic networks with  finite diameters.
This impossibility result is all the more striking given  the synchronization algorithm  proposed by Feldmann et 	
	al. (Algorithm 1 in~\cite{FelKS:spaa:20}) for the Beeping model, symmetric communication links, and 
	diffusive starts which does use bounded memory.
The authors leverage both symmetric communications and diffusive starts to build clocks
	whose absolute growths\footnote{%
	The absolute growth of $i$'s clock at round $t$ corresponds to the value of a virtual counter which represents 
	the absolute value by which $C_i$ has increased until round~$t$.} 
	differ from at most two for any pair of neighboring nodes.
Unfortunately, we do not know how to relax any of the latter  two assumptions.
Actually, we ignore whether synchronization is possible with bounded memory in static networks when
	asynchronous starts are arbitrary or when communications are not symmetric.
All these results and open questions are  summed up in Table~\ref{table:impossibility}.

\medskip
%%%%%%%%%%%%%%%%%%%%%%%%%%%%%%%%%%%%%%%%%%%
\begin{table}[h]%TODO Pourquoi il y a-t-il un `?` dans la case en bas à droite?
\centering
\begin{tabular}{|c|c|c|c|}
\hline
     & static, bidirectional & static & dynamic \\ [0.2cm]\hline

diffusive starts        & \multicolumn{1}{c|}{\checkmark~(\cite{FelKS:spaa:20})}&\multicolumn{1}{c|}{?}& \ding{55}~(Thm.~\ref{thm:no_dynamic_algo}) \\[0.2cm] \hline
asynchronous starts         & \multicolumn{1}{c|}{?}                   & \multicolumn{1}{c|}{?}        & \ding{55}~(Thm.~\ref{thm:no_dynamic_algo}) \\[0.2cm] \hline
self-stabilization      & \ding{55}                  &
    \multicolumn{1}{c|}{\ding{55}~(Thm.~\ref{thm:no_self_stab})} & \ding{55}~(Thm.~\ref{thm:no_self_stab} and \ref{thm:no_dynamic_algo}) \\[0.2cm] \hline
\end{tabular}
\caption{Clock synchronization problem with bounded memory.}
\label{table:impossibility}
\end{table}
%%%%%%%%%%%%%%%%%%%%%%%%%%%%%%%%%%%%%%%%%%%

\subsection{A lower bound for  clock synchronization }

From the  execution built in the  proof of Theorem~\ref{thm:no_dynamic_algo},
	we can obtain a lower bound on the number of states each node uses in  a synchronization algorithm
	that works in the network class ${\cal C}_{\leq n}$.
The key point in the proof is an upper bound on  $k^{!}\!$ when $k$ is a positive integer.

\begin{thm}\label{thm:lb}
Any  algorithm that achieves synchronization in the class 
	of dynamic networks with finite diameter and at most $n$ nodes requires each node to have 
	 at least equal to $  1 +  \left \lfloor \frac{1}{1,11} \log \left( \frac{n}{2} \right ) \right \rfloor $ states.
\end{thm}

\begin{proof}
Since $k^{!}\!$ divides $k!$, we have
	$$ \log k^{!}\! \leq \log k! = O(k \log k) .$$
In fact, a much stronger bound on 	$k^{!}\!$ holds:
Namely, as a consequence of the prime number theorem, one may show that 
	\begin{equation}\label{eq:equiv_pn}
	\log k^{!}\! \sim k .
	\end{equation}
Actually, we will content ourself with the following weaker but more explicit variant of Eq.~(\ref{eq:equiv_pn})
	\begin{equation}\label{eq:1.11}
	\log k^{!}\! \leq 1,11 \,  k .
	\end{equation}
To prove Eq.~(\ref{eq:1.11}), we first observe that $\log k^{!}\! $ admits the following expression:
	\begin{equation}\label{eq:logk^!}
	\log k^{!}\!  = \sum_{p } \left \lfloor \frac{\log k}{\log p} \right \rfloor \log p ,
	\end{equation}
	where the sum runs over the prime numbers $p$.
Indeed, for every prime number $p$, the power at which $p$ occurs in $ k^{!}\!  $ is equal to
	$$ \max \{ n \in \IN : p^n \leq k \} = \max  \{ n \in \IN : n \log p  \leq \log k  \} 
	                                                    =  \left \lfloor \frac{\log k}{\log p} \right \rfloor . $$
From Eq.~(\ref{eq:logk^!}), we derive the upper bound
	$$ \log   k^{!}\! \leq \sum_{p \leq k} \frac{\log k}{\log p} \log p \leq \pi(k) \log k , $$
	where  $\pi$ is the prime counting function 
	$\pi(k) \eqdef  | \{ p \mbox{ prime } : p \leq k\} |$.
Finally,   Eq.~(\ref{eq:1.11}) follows from   Chebyshev's estimate~\cite{HarW:book:39} 
	$$ \pi(k) \leq 1,11 \, \frac{k}{\log k} .$$     
An argument similar to that developed for Theorem~\ref{thm:lowerbounds}  and applied to the execution defined 
	in the proof of Theorem~\ref{thm:no_dynamic_algo} shows that the set of states ${\cal Q}$ of any 
	synchronization algorithm  in the network class ${\cal C}_{\leq n}$ satisfies
	$ | {\cal Q} |^{!}\!  > \frac{n}{2} $, and thus
	$$   | {\cal Q} | \geq 1 +  \left \lfloor \frac{1}{1,11} \log \left( \frac{n}{2} \right ) \right \rfloor \, .$$                        
\end{proof}

Unlike the similar lower bound established  for self-stabilizing algorithms in Section~\ref{sec:lowerboundsSS}, 
	we do not know whether the bound in Theorem~\ref{thm:lb} is  tight.
In other words,  we do not know what is the overhead induced by the self-stabilizing property for the problem
	of clock synchronization.

%\bibliographystyle{plain}
%\bibliography{../../Biblio/bibase}
%
%\paragraph{Acknowledgements.} 
%We would like to thank  Sylvain Gay and Vincent Peth for useful discussions and comments. 

\end{document}